\documentclass[10pt, final, journal, letterpaper, twocolumn]{IEEEtran}

\usepackage{amsfonts}
\usepackage[dvips]{graphicx}
\usepackage{epstopdf}
\usepackage{times}
\usepackage{cite}
\usepackage{amsmath}
\usepackage{array}
\usepackage{amssymb}

\newcommand{\bs}{\boldsymbol}

\usepackage{stfloats}
\usepackage{slashbox}
\usepackage{graphicx}
\usepackage{footnote}
\usepackage{booktabs}
\usepackage{array}
\usepackage{algorithmic}
\usepackage{algorithm}
\usepackage{subeqnarray}
\usepackage{cases}
\usepackage{threeparttable}
\usepackage{color}
\usepackage[normalsize]{subfigure}

\newtheorem{proposition}{Proposition}

\begin{document}

\title{Information and Energy Cooperation in OFDM Relaying: Protocols and Optimization}
\author{Yuan Liu,~\IEEEmembership{Member,~IEEE}, and Xiaodong Wang,~\IEEEmembership{Fellow,~IEEE}
\thanks{Y. Liu is with the School of Electronic and Information Engineering, South China University of Technology, Guangzhou, P. R. China. Email: eeyliu@scut.edu.cn.}
\thanks{X. Wang is with the Department of Electronic Engineering, Columbia University, New York, USA. Email: wangx@ee.columbia.edu.}
}

\maketitle

\vspace{-1.5cm}

\begin{abstract}
Integrating power transfer into wireless communications for supporting simultaneous wireless information and power transfer (SWIPT) is a promising technique in energy-constrained wireless networks. While most existing work on SWIPT focuses on capacity-energy characterization, the benefits of cooperative transmission for SWIPT are much less investigated. In this paper, we consider SWIPT in an orthogonal frequency-division multiplexing (OFDM) relaying system, where a source node transfers information and a fraction of power simultaneously to a relay node, and the relay node uses the harvested power from the source node to forward the source information to the destination. To support the simultaneous information and energy cooperation, we first propose a transmission protocol assuming that the direct link between the source and destination does not exist, namely power splitting (PS) relaying protocol, where the relay node splits the received signal power in the first hop into two separate parts, one for information decoding and the other for energy harvesting. Then, we consider the case that the direct link between the source and destination is available, and the transmission mode adaptation (TMA) protocol is proposed, where the transmission can be completed by cooperative mode and direct mode simultaneously (over different subcarriers). In direct mode, when the source transmits signal to the destination, the destination receives the signal as information and the relay node concurrently receives the signal for energy harvesting.
%
Joint resource allocation problems are formulated to maximize the system throughput. By using the Lagrangian dual method, we develop efficient algorithms to solve the nonconvex optimization problems.
\end{abstract}

\begin{keywords}
  Simultaneous wireless information and power transfer (SWIPT), orthogonal frequency-division multiplexing (OFDM), cooperative relay, resource allocation.
\end{keywords}

\section{Introduction}

\subsection{Motivation and Related Work}

Recently, simultaneous wireless information and power transfer (SWIPT) has become an emerging solution to prolong the lifetime of energy-constrained wireless networks. SWIPT refers to using the same wireless electromagnetic wave to transport energy for both information decoding and energy harvesting at the receiver. With SWIPT, the integrated transceiver designs can be exploited to efficiently use the available wireless resources.

In \cite{Varshney,Grover}, fundamental performance limits of SWIPT were studied from an information-theoretic perspective, and the tradeoff between the channel capacity and the harvested energy at the receiver was characterized. Such capacity-energy tradeoff was also derived for a multiuser system in \cite{Fouladgar}. These papers assume that the received energy can be still harvested after passing through an information decoder, which is not realizable yet due to practical electronic circuit limitations. A practical SWIPT receiver was designed in \cite{ZhouZhang}, where the received signal is split into two separate circuits, one for information decoding and the other for energy harvesting. The works in \cite{ZhangHo} and \cite{Xiang} studied SWIPT in multiple-input-multiple-output (MIMO) channels with perfect and imperfect channel state information (CSI), respectively. In particular, the authors in \cite{ZhangHo} proposed two transmission protocols, namely ``time switching" where the receiver switches between decoding information and harvesting energy, and ``power splitting" where the receiver splits the signal power into two parts for decoding information and harvesting energy. Note that the time switching protocol cannot operate information and power transfer simultaneously. The time switching protocol and the power splitting protocol with power allocation were investigated in \cite{LiuZhang1} and \cite{LiuZhang2}  for capacity-energy tradeoff, respectively.

A few attempts have been
made very recently to study wireless information and power transfer in cooperative relay systems \cite{Nasir,Gurakan}. To be specific, the outage probability and the ergodic capacity for time switching and power splitting protocols were derived in \cite{Nasir}. The authors in \cite{Gurakan} considered an energy cooperation scenario where each user shares a portion of the harvested energy with the other users, and the optimization of energy arrivals for throughput maximization using Lagrangian duality was developed. Note that \cite{Gurakan} considered a full-duplex relay which can transmit and receive data/energy simultaneously. In \cite{Ding}, the authors investigated wireless information and power transfer in cooperative networks where the randomly located relays assist one source-destination pair, and outage probability and diversity gain were characterized by stochastic geometry. The authors in \cite{Ding2} studied the outage performance of power strategies at an energy harvesting relay for multiple source-destination pairs. The optimal time-switching ratio was investigated in full-duplex relaying system under different communication modes in \cite{Zhong}.

On the other hand, orthogonal frequency-division multiplexing (OFDM) based relaying transmission is a powerful tool to enable high date rates and has been adopted as a key technique for the next generation communications.
The optimization and resource allocation schemes were proposed for various settings in OFDM relay systems, e.g., single-user single-relay \cite{Hammerstrom,Louveaux,Hsu}, multi-user single relay \cite{Hajiaghayi}, and multi-user multi-relay \cite{Kwak,Ng,YuanTWC10,YuanTCOM12,YuanTWC13,Tao2013}. For OFDM-based SWIPT, power optimization in OFDMA systems for different configurations was studied in \cite{HuangLarsson}. The authors in \cite{Kwan} studied SWIPT in downlink OFDMA systems with power splitting receivers, and suboptimal iterative algorithms were developed to solve the non-convex resource allocation problems. In \cite{XunZhou2014}, the weighted sum-rate maximization problem subject to a minimum harvested energy constraint on each user was studied, under either time switching or power splitting protocol to coordinate energy harvesting and information decoding.

The above OFDM-based SWIPT \cite{HuangLarsson,Kwan,XunZhou2014} focus on traditional direct transmission. The authors in \cite{Xiong} studied SWIPT in MIMO-OFDM relay networks and solved the non-convex resource allocation problems in a suboptimal step-wise way. It is thus observed that SWIPT in OFDM relay systems is not well studied. Using the architecture of OFDM relaying, the potential of SWIPT can be fully explored in various dimensions, including frequency diversity and cooperative diversity by designing efficient resource allocation schemes. It is thus of great importance and necessity to capture the full potential of OFDM relaying for SWIPT. This motivates our study. In this paper, we aim to answer the following questions: What is the transmission protocol to be used for implementing SWIPT in an existing OFDM relay network? How to efficiently allocate wireless resources, such as power and subcarriers to maximize the throughput given the  transmission protocols?


\subsection{Contributions}

In this paper, we consider a three-node relaying scenario, where a source node transfers a portion of its energy to a decode-and-forward (DF) relay node in return for its assistance in information transmission using OFDM. The main contributions and results of this paper are summarized as follows:
\begin{itemize}
  \item Assuming the direct link between the source and destination is not available, power splitting relaying (PS) protocol for supporting SWIPT in OFDM relaying is considered, where the relay node splits the received signal power in the first hop into two separate parts, one for information decoding and the other for energy harvesting. Then considering the case that the direct link between the source and destination is available, we propose the transmission mode adaptation (TMA) protocol to support SWIPT in OFDM relaying, where the end-to-end transmission can be completed through relay (cooperative mode) and direct source-to-destination transmission (direct mode) simultaneously. For direct mode, when the source transmits signal to the destination, the destination receives the signal as information and the relay node concurrently receives the signal for energy harvesting.
      For both protocols, the relay node harvests power from the source node as the transmit power for information forwarding, which can be regarded as \emph{simultaneous} two-level cooperation, i.e., information-level cooperation and energy-level cooperation, thanks to the parallel structure of OFDM relaying.
  \item For the PS protocol, the joint problem of power allocation and determining power splitting ratio is studied. For the TMA protocol, we consider a joint problem of mode selection (between cooperative and direct modes), subcarrier and power allocation. By using the Lagrangian duality method, we develop efficient algorithms to solve these nonconvex problems.
  \item Simulation results show that for the PS protocol, more resources are allocated to information transfer and less resources for power transfer.  The power allocated to information transfer is increasing when the relay node is closer to the destination; on the other hand, the power allocated to power transfer is decreasing in this case. For TMA protocol, more power should be allocated to cooperative mode when the relay node is located at the middle, and less power is located when the relay node is closer to either the source or destination. The case is reversed for the direct mode. We also show that for both protocols, it is better to place the relay node closer to the source node rather than in the middle between the source and the destination.
\end{itemize}

\subsection{Organization}

The remainder of this paper is organized as follows. Section II describes the system model and the proposed transmission protocols. Section III and Section IV present the joint resource allocation problems for the proposed PS and TMA protocols, respectively. Section V provides extensive simulations. Finally, we conclude the paper in Section VI.

\section{System Model and Transmission Protocols}

\begin{figure}[t]
\begin{centering}
\includegraphics[scale=0.8]{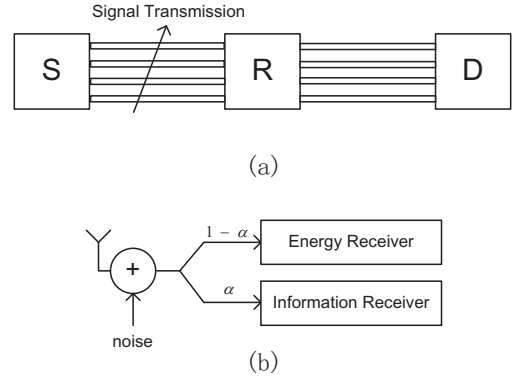}
\vspace{-0.1cm}
 \caption{(a) Illustration of the PS relaying protocol. (b) Diagram of the relay receiver over all subcarriers.}\label{fig:model_PS}
\end{centering}
\vspace{-0.3cm}
\end{figure}
\begin{figure}[t]
\begin{centering}
\includegraphics[scale=0.8]{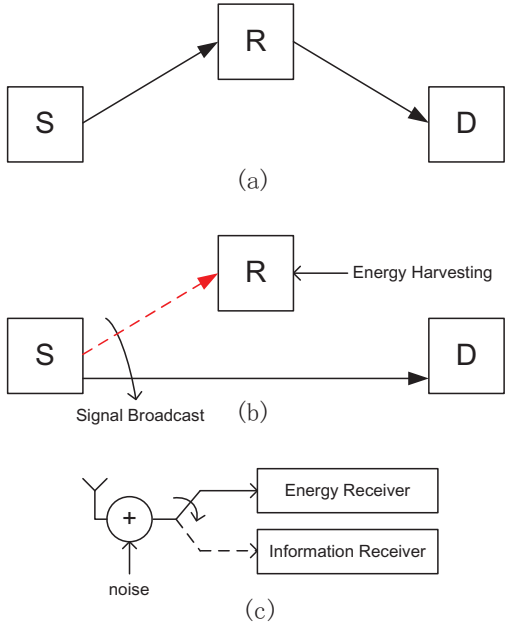}
\vspace{-0.1cm}
 \caption{Illustration of the TMA relaying protocol. (a) Cooperative information transmission mode. (b) Direct information transmission mode. (c) Diagram of the relay receiver.}\label{fig:model_MA}
\end{centering}
\vspace{-0.3cm}
\end{figure}

We consider a three-node relaying system, consisting of a source node $S$, a relay node $R$, and a destination node $D$ as shown in Figs. \ref{fig:model_PS}-\ref{fig:model_MA}. OFDM is used at the physical layer for the two-hop transmission. The relay node is half-duplex, i.e., it cannot receive and transmit simultaneously. The end-to-end transmission is completed by two equal time slots. In the first time slot, the source node transmits signals over OFDM channels while the relay node receives. In the second time slot, the relay node forwards the processed information to the destination using the DF strategy.

Using the parallel structure of OFDM relaying, we first propose the PS transmission protocol to support SWIPT assuming that the direct link between the source and destination does not exist.\footnote{This is due to the impairments such as wireless fading, shadowing, path loss, and obstacles, which make the effects of direct link negligible.}
Fig. \ref{fig:model_PS} depicts the proposed PS protocol.
The signal power on \emph{all} subcarriers is split into two parts at the relay receiver, one for information decoding and one for energy harvesting. Specifically,  the relay receiver determines the optimal power splitting ratio $\alpha$ allocated to the information decoder and $1-\alpha$ to the energy harvester. That is, all subcarriers have the same power splitting ratio. The energy harvested during the first time slot is
used as a source of power for the relay to forward the source information to the destination.

Then, we assume that the direct link between the source and destination is available and propose the TMA protocol in Fig. \ref{fig:model_MA}. By exploring the direct link between the source and destination, the end-to-end information transmission can be completed by two transmission modes simultaneously (over different subcarriers), e.g., cooperative mode (see Fig. \ref{fig:model_MA}(a)) and direct mode (see Fig. \ref{fig:model_MA}(b)). Correspondingly, the OFDM subcarriers in the first time slot are partitioned into two groups for the two transmission modes respectively. In particular, for the direct mode, when the source transmits signal to the destination, the destination receives the signal as information and the relay node concurrently receives the signal for energy harvesting. The harvested energy at the relay node at the subcarriers with direct mode is used as a source of power to forward the source information for cooperative mode. At the relay node (see Fig. \ref{fig:model_MA}(c)), the signal power from the subcarriers with cooperative mode and direct mode are sent to information receiver and energy receiver, respectively.

There are two points about the protocols to be noted. First, for both protocols, the harvested powers from all subcarriers during the first time slot are added up as a total power for information forwarding over all subcarriers in the second time slot. Second, we assume that the information from one group of subcarriers in the first hop can be decoded and re-encoded jointly and then transmitted over a different group of subcarriers in the next hop. This is optimal for DF relaying over multi-channels \cite{Liang}.

It is noticed that for the PS protocol, we let the relay receiver split the signal energy into two parts for all subcarriers with the same power splitting ratio, since the power splitting is performed in analog domain before the digital OFDM demodulation. Moreover, for the TMA protocol, we assume that the relay receiver has ideal bandpass filters so that it is able to tap into different subcarriers. This is also possible in OFDM-based ultra wide band (UWB) system where the ratio of bandwidth and carrier frequency is more than 20\% (according FCC's definition), which makes the bandpass filters not sharp and thus feasible.

The fading is modeled by large-scale path loss, shadowing, and small-scale frequency-selective Rayleigh fading. The transmission from the source to destination is divided into consecutive frames. It is assumed that the fading remains unchanged within each transmission frame but varies from one frame to another. We also assume that perfect CSI is available at all nodes. Without loss of generality, the additive noises at all nodes are independent circular symmetric complex Gaussian random variables, each having zero mean and unit variance.

In this paper, we assume that energy loss and nonlinearity, due to processing and circuitry at the transmit/receive, are neglected for simplicity. The assumption is widely adopted in the prior work \cite{ZhouZhang,ZhangHo,Xiang,LiuZhang1,LiuZhang2,Nasir,Gurakan,XunZhou2014}.

%

\section{Power Splitting Protocol}

In this section, we study the PS protocol and the corresponding joint optimization problem of power allocation and power splitting.

\subsection{Problem Formulation}

Let $p_{s,n}$ denote the source power on subcarrier $n$ in the first time slot.
As stated in \cite{XunZhou2014}, all subcarriers would have the same power splitting ratio in an OFDM-based SWIPT system. Thus we denote $\alpha$ as the fraction of power for information decoding  and the remaining $1-\alpha$ is the fraction of power for energy harvesting.
%
The harvested power from the source node over all subcarriers in the first time slot is used as the source of transmit power to forward the source information to the destination. Let $h_n$ and $g_n$ are the channel gains over subcarrier $n$ in the first and second hops, respectively. Denote $\mathcal N:=\{1,\cdots,N\}$ as the set of subcarriers, then the transmit power of the source $\{p_{s,n}\}$ and the relay $\{p_{r,n}\}$ should satisfy:
\begin{equation}\label{eqn:ps'}
  \sum_{n\in\mathcal N}p_{s,n}\leq P,
\end{equation}
and
\begin{equation}\label{eqn:pr'}
  \sum_{n\in\mathcal N}p_{r,n}\leq\sum_{n\in\mathcal N}p_{s,n}h_n(1-\alpha).
\end{equation}
In addition, we consider the peak power constraint on each subcarrier\footnote{This is known as ``spectral mask" in OFDM systems introduced by FCC, for considering practical power-amplifier limits and controlling interference to other devices.}:
\begin{eqnarray}\label{eqn:mask2}
  \mathcal P_{{\rm PS}}:=\{0\leq p_{s,n}\leq \overline{p}_{n},
  0\leq p_{r,n}\leq \overline{p}_{n},\forall n\}.
\end{eqnarray}
To avoid trivial solution, we assume that $\sum_{n\in\mathcal N}{\bar p}_n\geq P$.

We can easily obtain the end-to-end achievable rate of the PS protocol as
\begin{align}\label{eqn:RPS}
  R_{{\rm PS}} = \frac{1}{2}\min\Big\{\sum_{n\in\mathcal N}\log(1+p_{s,n}h_n\alpha),
  \sum_{n\in\mathcal N}\log(1+p_{r,n}g_n)\Big\}.
\end{align}

In the PS protocol, our goal is to maximize the system throughput by jointly determining the transmit power and power splitting ratio on each subcarrier. Let $\bs p=\{p_{s,n},p_{r,n}\}$, and the problem can be formulated as
\begin{subequations}\label{eqn:p3}
\begin{align}
\textbf{P1}:~~\max_{\{\bs p, \alpha\}}~&R_{{\rm PS}} \\
{\rm s.t.}~~ &\eqref{eqn:ps'}-\eqref{eqn:RPS}.
\end{align}
\end{subequations}

\subsection{Proposed Solution}

The problem in \textbf{P1} is a mixed integer optimization problem and nonconvex. Finding the optimal solution needs exhaustive search with prohibitive complexity. However, as shown in \cite{Yu}, for the nonconvex optimization problems in multicarrier systems, the duality gap is zero under the ``time-sharing" condition, and the time-sharing condition is always satisfied if the number of subcarriers is large. Therefore, in this paper we solve the resource allocation problems in dual domain.

We first introduce non-negative Lagrangian multipliers $\mu$, $\nu$, $\lambda$, and $\kappa$ rated with the constraints \eqref{eqn:ps'}-\eqref{eqn:RPS}, respectively. The Lagrangian of \textbf{P1} is
\begin{align}\label{eqn:La-PS0}
  \mathcal L_{{\rm PS}}(\lambda,\kappa,\mu,\nu,\bs p,\alpha)=\lambda\left[\frac{1}{2}\sum_{n\in\mathcal N}\log(1+p_{s,n}h_n\alpha)-R_{{\rm PS}}\right]\nonumber\\
  +\kappa\left[\frac{1}{2} \sum_{n\in\mathcal N}\log(1+p_{r,n}g_n)-R_{{\rm PS}}\right]+ R_{{\rm PS}}\nonumber\\
  +\mu\left[P-\sum_{n\in\mathcal N}p_{s,n}\right]
  +\nu\left[\sum_{n\in\mathcal N}p_{s,n}h_n(1-\alpha)-\sum_{n\in\mathcal N}p_{r,n}\right].
\end{align}

Define $\mathcal D_1$ as the set of all primal variables that satisfy the constraints, then the dual function of \textbf{P1} can be expressed as
\begin{equation}
  g_{{\rm PS}}(\lambda,\kappa,\mu,\nu)=\max_{\{\bs p,\alpha\}\in\mathcal D_1}\mathcal L_{{\rm PS}}(\lambda,\kappa,\mu,\nu,\bs p,\alpha).
\end{equation}

Computing the dual function $g_{\rm PS}(\lambda,\kappa,\mu,\nu)$ needs to find the optimal $\{\bs p^*,\alpha^*\}$
for the given dual variables $\{\lambda,\kappa,\mu,\nu\}$. In the following, we detail the derivations.

We first investigate the part of the dual function with respect to the rate variable $R_{{\rm PS}}$:
\begin{equation}\label{eqn:g0}
  g_0(\lambda,\kappa)=\max_{R_{{\rm PS}}}~(1-\lambda-\kappa)R_{{\rm PS}}.
\end{equation}

To ensure the dual function to be bounded, we have $1-\lambda-\kappa=0$ such that $g_0(\lambda,\kappa)=0$, which means that $\kappa=1-\lambda$. It is nontrivial that $0\leq\lambda\leq1$ making $\kappa$ non-negative.

By substituting the result to \eqref{eqn:La-PS0}, the Lagrangian becomes
\begin{align}\label{eqn:La-PS}
 \mathcal L_{{\rm PS}}(\lambda,\mu,\nu,\bs p,\alpha)=\frac{\lambda}{2}\sum_{n\in\mathcal N}\log(1+p_{s,n}h_n\alpha) \nonumber\\
 +\frac{1-\lambda}{2}\sum_{n\in\mathcal N}\log(1+p_{r,n}g_n)
 +\mu\left[P-\sum_{n\in\mathcal N}p_{s,n}\right] \nonumber\\
 +\nu\left[\sum_{n\in\mathcal N}p_{s,n}h_n(1-\alpha)-\sum_{n\in\mathcal N}p_{r,n}\right].
\end{align}

Note that when tackling the problem, the power splitting ratio $\alpha$ appears in the achievable rate expression and couples the power allocation on each subcarrier, which results in a nonconvex problem. To this end, we first derive the optimal power allocations for given $\alpha$ in the following proposition.

\begin{proposition}\label{pro:pps}
  For a given power splitting ratio $\alpha$, the optimal power allocations are
  \begin{eqnarray}
  p_{s,n}^*=\begin{cases}
\overline{p}_{n},&{\rm if}~\nu h_n(1-\alpha)\geq\mu\\
  \left[\frac{\lambda}{2(\mu-\nu h_n(1-\alpha))}-\frac{1}{h_n\alpha}\right]
  ^{\overline{p}_{n}}_0,&{\rm otherwise}.
    \end{cases}
  \end{eqnarray}
  \begin{eqnarray}
  p_{r,n}^*=\left[\frac{1-\lambda}{2\nu}-\frac{1}{g_n}\right]^{\overline{p}_{n}}_0.
  \end{eqnarray}
\end{proposition}

\begin{proof}
Note that for a fixed $\alpha$, the Lagrangian is concave in the power allocation variables. For each $p_{s,n}$, the corresponding sub-Lagrangian is
\begin{equation}
  L_{{\rm PS}}(p_{s,n})=\frac{\lambda}{2}\log(1+p_{s,n}h_n\alpha)-p_{s,n}(\mu-\nu h_n(1-\alpha)).
\end{equation}
Since $\lambda\geq0$, $L_{{\rm PS}}(p_{s,n})$ is an increasing function of $p_{s,n}$ if $\nu h_n(1-\alpha)\geq\mu$ and thus the optimal solution is $p_{s,n}^*=\overline{p}_{n}$. Otherwise $L_{{\rm PS}}(p_{s,n})$ is a concave function of $p_{s,n}$. By the Karush-Kuhn-Tucker (KKT) conditions \cite{Boyd}, the maximum is achieved at $\frac{\partial L_{{\rm PS}}(p_{s,n})}{\partial p_{s,n}}=0$. $p_{r,n}^*$ directly follows the KKT conditions.
\end{proof}

Substituting the optimal power allocations into the Lagrangian, it is a function of the power splitting ratio $\alpha$, which still turns out to be neither a convex nor a concave function.
As the closed-form solution of the power splitting ratio $\alpha$ seems infeasible, we resort to a numerical method, i.e., we enumerate all possible values of $\alpha$ in the interval $[0, 1]$ with a sufficiently small step-size and select the value of $\alpha^*$ that achieves the maximum of the Lagrangian.

After finding the optimal $\{\bs p^*,\alpha^*\}$, we turn to the dual problem, which can be expressed as
\begin{subequations}
\begin{align}
\min_{\{\lambda,\mu,\nu\}}~&g_{\rm PS}(\lambda,\mu,\nu) \\
{\rm s.t.}~~&0\leq\lambda\leq1,\mu\geq0,\nu\geq0.
\end{align}
\end{subequations}

Since a dual function is always convex by definition, we employ the ellipsoid method \cite{Boyd} to simultaneously update the dual variables $\{\lambda,\mu,\nu\}$ toward the optimal  $\{\lambda^*,\mu^*,\nu^*\}$ by using the following subgradients.
\begin{equation}\label{eqn:delta-lambda0}
  \Delta\lambda = \frac{1}{2}\sum_{n\in\mathcal N}\log(1+p_{s,n}^*h_n\alpha^*)-\frac{1}{2}\sum_{n\in\mathcal N}\log(1+p_{r,n}^*g_n),
\end{equation}
\begin{equation}\label{eqn:delta-mu0}
  \Delta\mu=P-\sum_{n\in\mathcal N}p_{s,n}^*,
\end{equation}
\begin{equation}\label{eqn:delta-nu0}
  \Delta\nu=\sum_{n\in\mathcal N}p_{s,n}^*h_n(1-\alpha^*)-\sum_{n\in\mathcal N}p_{r,n}^*.
\end{equation}

\subsection{Algorithm and Complexity}
\begin{algorithm}[!t]
\caption{Proposed Algorithm for \textbf{P1}}
\begin{algorithmic}[1]
\STATE \textbf{initialize} $\{\lambda,\mu,\nu\}$ as non-negative values.
\REPEAT \STATE Compute the optimal power allocations $\bs p^*(\alpha,\lambda,\mu,\nu)$ using Proposition \ref{pro:pps}. \STATE Find the optimal power splitting ratio $\alpha^*(\lambda,\mu,\nu)$ using numerical search with step-size $\epsilon$. Here $0\leq\epsilon\ll1$ controls accuracy. \STATE Update $\{\lambda,\mu,\nu\}$ by the ellipsoid method using the subgradients defined in \eqref{eqn:delta-lambda0}-\eqref{eqn:delta-nu0}.
 \UNTIL{$\{\lambda,\mu,\nu\}$ converge.}
\end{algorithmic}
\end{algorithm}

We summarize the proposed solution for \textbf{P1} in Algorithm 1. The computational complexity of the ellipsoid method is $\mathcal{O}(q^2)$, where $q$ is the number of the dual variables, and clearly $q=3$ in our paper. Denote $\epsilon$ as the step-size of the search procedure for the optimal $\alpha^*$, the total complexity of the proposed algorithm is $\mathcal{O}(q^2/\epsilon)$.

%

\section{Mode Adaptation Protocol}

In this section, we study the TMA protocol and the corresponding joint optimization problem of transmission mode selection, subcarrier assignment, and power allocation.

\subsection{Problem Formulation}

Let $p_{sr,n}$ and $p_{sd,n}$ denote the transmit power of the source-to-relay and source-to-destination links on subcarrier $n$ at the first time slot, respectively. $f_n$ represents the channel gain of the source-to-destination link on subcarrier $n$. Binary variables $y_{c,n}$ and $y_{d,n}$ are used for subcarrier assignment at the first time slot, $y_{c,n}=1$ indicates that subcarrier $n$ is used for cooperative mode and $y_{c,n}=0$ otherwise; $y_{d,n}=1$ indicates that subcarrier $n$ is used for direct mode and $y_{d,n}=0$ otherwise. $y_{c,n}$ and $y_{d,n}$ satisfy the following constraint
\begin{equation}\label{eqn:y}
  y_{c,n} + y_{d,n}\leq1,~\forall n\in\mathcal N.
\end{equation}
Then, the transmit power of the source should satisfy the total power constraint:
\begin{equation}\label{eqn:ps3}
  \sum_{n\in\mathcal N}\left(y_{c,n}p_{sr,n}+y_{d,n}p_{sd,n}\right)\leq P.
\end{equation}

The harvested power at the relay node on subcarrier $n$ during the direct mode is $y_{d,n}p_{sd,n}h_n$. The total harvested power is used as a source of power $\{p_{r,n}\}$ to forward information of the source node and satisfy
\begin{equation}\label{eqn:pr3}
  \sum_{n\in\mathcal N}p_{r,n}\leq\sum_{n\in\mathcal N}y_{d,n}p_{sd,n}h_n.
\end{equation}
The power allocations are also constrained by
\begin{eqnarray}\label{eqn:mask3}
  \mathcal P_{{\rm TMA}}:=\{0\leq p_{sr,n}\leq \overline{p}_{n},0\leq p_{sd,n}\leq \overline{p}_{n},\nonumber\\
  0\leq p_{r,n}\leq \overline{p}_{n},\forall n\}.
\end{eqnarray}

The end-to-end transmission rate of the TMA protocol is the sum rates of the direct and cooperative modes:
\begin{align}\label{eqn:RTM}
  R_{{\rm TMA}} = R_d + R_c,
\end{align}
where
\begin{equation}
R_d=\sum_{n\in\mathcal N}y_{d,n}\log(1+p_{sd,n}f_n),
\end{equation}
and
\begin{align}\label{eqn:Rc}
R_c=\frac{1}{2}\min\Big\{\sum_{n\in\mathcal N}y_{c,n}\log(1+p_{sr,n}h_n),\nonumber\\
  \sum_{n\in\mathcal N}y_{c,n}\log(1+p_{r,n}g_n+p_{sr,n}f_n)\Big\}.
\end{align}
Note that in \eqref{eqn:Rc}, we assume that the destination can receive the information from the source to the relay for maximal-ratio combining.

In the TMA protocol, the goal is to maximize the system throughput by jointly allocating power, subcarriers, and transmission modes. Denote $\bs p=\{p_{sr,n},p_{sd,n},p_{r,n}\}$ and $\bs y=\{y_{c,n},y_{d,n}\}$, this problem can be formulated as
\begin{subequations}\label{eqn:p3}
\begin{align}
\textbf{P2}:~~\max_{\{\bs p,\bs y\}}~&R_{{\rm TMA}} \\
{\rm s.t.}~~ &\eqref{eqn:y}-\eqref{eqn:RTM}.
\end{align}
\end{subequations}

\subsection{Proposed Solution}

The problem in \textbf{P2} is a mixed integer optimization problem and nonconvex. To make it more tractable, we first relax the binary variables to real ones, i.e., $0\preceq\bs y\preceq1$, and defining new variables $\bs t =\{t_{c,n},t_{c,n}',t_{d,n}\}$ with $t_{c,n}=y_{c,n}p_{sr,n}$, $t_{c,n}'=y_{c,n}p_{r,n}$ and $t_{d,n}=y_{d,n}p_{sd,n}$.

The continuous relaxation makes $\bs t$ the time sharing factors. Moreover, after the relaxation, \textbf{P2} is a convex problem. We notice that applying conventional software packages directly on \textbf{P2} is not sufficient to solve the original problem \textbf{P2}, since there is no guarantee that they can return the binary $\bs y$.

We can write the  Lagrangian as
\begin{align}\label{eqn:La-TMA}
 &\mathcal L_{{\rm TMA}}(\beta,\gamma,\delta,\bs p,\bs y, \bs t)=\sum_{n\in\mathcal N}y_{d,n}\log\left(1+\frac{t_{d,n}f_n}{y_{d,n}}\right) \nonumber\\
 &+\frac{\beta}{2}\sum_{n\in\mathcal N}y_{c,n}\log\left(1+\frac{t_{c,n}h_n}{y_{c,n}}\right)\nonumber\\
 &+\frac{1-\beta}{2}\sum_{n\in\mathcal N}y_{c,n}\log\left(1+\frac{t_{c,n}'g_n+t_{c,n}f_n}{y_{c,n}}\right)\nonumber\\
 &+\gamma\left[P-\sum_{n\in\mathcal N}\left(t_{c,n}+t_{d,n}\right)\right]
 +\delta\left[\sum_{n\in\mathcal N}t_{d,n}h_n-\sum_{n\in\mathcal N}t_{c,n}'\right],
\end{align}
where $\beta$, $\gamma$, and $\delta$ are non-negative Lagrangian multipliers relating  to the constraints \eqref{eqn:RTM}, \eqref{eqn:ps3}, and \eqref{eqn:pr3}, respectively. Moreover, the constraint $0\leq\beta\leq1$ must hold. The dual function can be expressed as
\begin{equation}
  g_{{\rm TMA}}(\beta,\gamma,\delta)=\max_{\{\bs p,\bs y, \bs t\}\in\mathcal D_2} \mathcal L_{{\rm TMA}}(\beta,\gamma,\delta,\bs p,\bs y, \bs t),
\end{equation}
where $\mathcal D_2$ as the set of all primal variables that satisfy the constraints.

We first get the optimal power allocations in the following proposition.
\begin{proposition}\label{prop:p3'}
  The optimal power allocations of the relaxed \textbf{P2} are
\begin{equation}\label{eqn:psr-opt}
    t_{c,n}^*=y_{c,n}p_{sr,n}^*=y_{c,n}\left[\frac{\beta}{2(\gamma-\delta)}-\frac{1}{h_n}\right]
  ^{\overline{p}_{n}}_0,
  \end{equation}
  \begin{equation}\label{eqn:pr3-opt}
  t_{c,n}'^*=y_{c,n}p_{r,n}^*= y_{c,n}\left[\frac{(1-\beta)f_n}{2\delta g_n}-\frac{p_{sr,n}^*f_n}{g_n}-\frac{1}{g_n}\right]
  ^{\overline{p}_{n}}_0,
  \end{equation}
  and $t_{d,n}^*=y_{d,n}p_{sd,n}^*$ with
  \begin{eqnarray}
    p_{sd,n}^*=\begin{cases}
  \left[\frac{1}{\gamma-\delta h_n}-\frac{1}{f_n}\right]
  ^{\overline{p}_{n}}_0,~&{\rm if}~\gamma\geq\delta h_n\\
  \overline{p}_{n},~&{\rm otherwise},
    \end{cases}
  \end{eqnarray}
  %
  %
\end{proposition}

\begin{proof}
  \eqref{eqn:psr-opt} and \eqref{eqn:pr3-opt} can be easily obtained by applying KKT conditions.
  The Lagrangian with respect to $p_{sd,n}$ can be written as
  \begin{equation}
   L_{d,n}(p_{sd,n})=\log(1+p_{sd,n}f_n)-p_{sd,n}(\gamma-\delta h_n).
  \end{equation}
  Since $L_{d,n}$ is an increasing function of $p_{sd,n}$ if $\gamma<\delta h_n$, thus the maximum of $L_{d,n}$ is at $p_{sd,n}^*=\overline{p}_{n}$. If $\gamma\geq\delta h_n$, $L_{d,n}$ is a concave function of $p_{sd,n}$ and its maximum is at $\frac{\partial L_{d,n}(p_{sd,n})}{\partial p_{sd,n}}=0$.
\end{proof}

We are now ready to find the optimal subcarrier assignment between cooperative mode and direct mode at the first time slot  by the following proposition.

\begin{proposition}\label{pro:mode}
For each subcarrier $n$ in the first hop, the optimal subcarrier assignment is given by
\begin{eqnarray}\label{eqn:y-opt}
  \begin{cases}
    y_{c,n}^*=1~{\rm and}~y_{d,n}^*=0,{\rm if}~J_{c,n}\geq J_{d,n}\\
    y_{c,n}^*=0~{\rm and}~y_{d,n}^*=1,{\rm otherwise,}
  \end{cases}
\end{eqnarray}
where
\begin{align}
J_{c,n}=&\frac{\beta}{2}\log(1+p_{sr,n}^*h_n)+\frac{1-\beta}{2}\log(1+p_{sr,n}^*f_n+p_{r,n}^*g_n)\nonumber\\
&-p_{sr,n}^*\gamma-p_{r,n}^*\delta,
\end{align}
\begin{equation}
J_{d,n}=\log(1+p_{sd,n}^*f_n)-p_{sd,n}^*(\gamma-\delta h_n).
\end{equation}
Here $p_{sr,n}^*$, $p_{sr,n}^*$ and $p_{sd,n}^*$ are obtained in Proposition \ref{prop:p3'}.
\end{proposition}

\begin{proof}
To maximize the Lagrangian $\mathcal L_{{\rm TMA}}$, the Lagrangian over each subcarrier should be maximized, which can be expressed by
\begin{eqnarray}
    \max_{y_{c,n},y_{d,n}\geq0} &&y_{c,n}J_{c,n} + y_{d,n}J_{d,n}\\
    {\rm s.t.}~~~~&&y_{c,n}+y_{d,n}\leq1.
\end{eqnarray}
Since the objective function has a bounded objective and its maximization over $y_{c,n}$ and $y_{d,n}$ is a linear programming problem, a global optimum can be found at the extreme points of the feasible region \cite{Garfinkel}. Thus we can determine the optimal binary solutions of $y_{c,n}$ and $y_{d,n}$ by exhaustive search over $\{0,1\}$. In other words, we let the maximizer of $J_{c,n}$ and $J_{d,n}$ be active and the remaining one be inactive in each subcarrier $n$ for the exclusive subcarrier assignment.
\end{proof}

It is worth noting that if $J_{c,n}=J_{d,n}$, there exist infinite maximizers of the objective function, but we are only interested in the optimal solution in binary form to recover the primal exclusive subcarrier assignment constraint \eqref{eqn:y}. The binary recovery for the case of $J_{c,n}=J_{d,n}$ may result in that the optimal power allocations obtained in dual domain are not feasible for the primal problem. However, the probability of $J_{c,n}=J_{d,n}$ is zero, since $J_{c,n}$ and $J_{d,n}$ are the functions of $h_n$ and $f_n$ which are the independent continuous random variables. Hence, we can uniquely recover the optimal binary variables $\bs y^*$ for given optimal $\{\beta^*,\gamma^*,\delta^*\}$ with probability 1.

After determining the optimal primal variables $\{\bs p^*,\bs y^*,\bs t^*\}$, the dual problem can be expressed as
\begin{subequations}
\begin{align}
\min_{\{\beta,\gamma,\delta\}}~&g_{\rm TMA}(\beta,\gamma,\delta) \\
{\rm s.t.}~~&0\leq\beta\leq1,\gamma\geq0,\delta\geq0.
\end{align}
\end{subequations}

The optimal dual variables $\{\beta^*,\gamma^*,\delta^*\}$ can be found by using the ellipsoid method according to the following subgradients:
\begin{align}\label{eqn:delta-beta}
  \Delta\beta=&\frac{1}{2}\sum_{n\in\mathcal N}y_{c,n}^*\log\left(1+\frac{t_{sr,n}^*h_n}{y_{c,n}^*}\right)\nonumber\\
  &-\frac{1}{2}\sum_{n\in\mathcal N}y_{c,n}^*\log\left(1+\frac{t_{c,n}'^*g_n+t_{c,n}^*f_n}{y_{c,n}^*}\right),
\end{align}
\begin{equation}\label{eqn:delta-gamma}
  \Delta\gamma=P-\sum_{n\in\mathcal N}\left(t_{sr,n}^*+t_{sd,n}^*\right),
\end{equation}
\begin{equation}\label{eqn:delta-delta}
  \Delta\delta=\sum_{n\in\mathcal N}t_{sd,n}^*h_n-\sum_{n\in\mathcal N}t_{c,n}'^*.
\end{equation}

\subsection{Algorithm and Complexity}
Finally, we present the proposed solution for \textbf{P2} in Algorithm 2. The complexity in Proposition \ref{pro:mode} for all subcarriers is $\mathcal{O}(2N)$. Combining the complexity of dual updating, the total computational complexity is $\mathcal{O}(2Nq^2)$.
\begin{algorithm}[!t]
\caption{Proposed Algorithm for \textbf{P2}}
\begin{algorithmic}[1]
\STATE \textbf{initialize} $\{\beta,\gamma,\delta\}$ as non-negative values.
\REPEAT \STATE Compute the optimal power allocations $\bs p^*(\beta,\gamma,\delta)$ using Proposition \ref{prop:p3'}. \STATE Find the optimal subcarrier assignment $\bs y^*(\beta,\gamma,\delta)$ using \eqref{eqn:y-opt}. \STATE Update $\{\beta,\gamma,\delta\}$ by the ellipsoid method using the subgradients defined in \eqref{eqn:delta-beta}-\eqref{eqn:delta-delta}.
 \UNTIL{$\{\beta,\gamma,\delta\}$ converge.}
\end{algorithmic}
\end{algorithm}

\section{Simulation Results}

In this section, we conduct comprehensive simulations to evaluate the performance of the proposed protocols.

We set the distance between the source $S$ and the destination $D$ is $2$, and the relay node $R$ is located in a line between the source $S$ and the destination $D$. The
Stanford University Interim (SUI)-6 channel model \cite{SUI} is employed for generating OFDM channels and the path-loss exponent is set as $3.5$. The standard deviation of lognormal shadowing is set to be 5.8 dB and the small-scale fading is modeled by Rayleigh fading. The number of the subcarriers is $64$. Moreover, we set $\overline{p}_n=5$dB, $\forall n$. The step-size of Algorithm 1 for searching $\alpha^*$ is set to be $\epsilon=10^{-3}$. In this section, all numerical results are obtained by averaging 1000 channel realizations.

\subsection{Performance of Symmetric Relay}

We first introduce the benchmarks as follows.
\begin{enumerate}
  \item \emph{Performance Upper Bound of Power Splitting (UB-PS):} An upper bound of the PS protocol can be obtained by assuming that the receive of the relay can dynamically determine the power splitting ratio on each subcarrier, i.e., $\alpha_n$ for all $n\in\mathcal N$ instead of $\alpha$. Although the scheme may not be implemented in practice, it actually provides a performance upper bound for the PS protocol. We give the derivations in Appendix \ref{app:UB}.
  \item \emph{Equal Power Splitting (EPS):} For the EPS scheme, the power splitting ratio is set to be $\alpha=\frac{1}{2}$, and the pure power allocation problem is studied using the algorithm for PS protocol. The complexity is $\mathcal{O}(q^2)$.
  \item \emph{Equal Transmission Modes (ETM):} For the ETM scheme, the half of subcarriers in the first time slot is used for cooperative mode and half for direct mode, and the pure power allocation problem is studied using the algorithm for TMA protocol. The complexity is $\mathcal{O}(q^2)$.
\end{enumerate}

In this subsection, we consider the case when relay is located the middle of the source $S$ and the destination $D$. From Fig. \ref{fig:sum-rate}, we observe that the TMA protocol outperforms the PS protocol as expected. The performance of the PS protocol is very close to the upper bound. This suggests that splitting all subcarriers with an identical power ratio is very efficient. In addition, we observe that as SNR increases, the performance of all schemes are bounded at about SNR=25dB, due to the peak power constraint on each subcarrier. Moreover, the gains over the benchmarks illustrate the benefits of adaptive subcarrier assignment and power splitting.

From Fig. \ref{fig:sub_num}, it is found that for the TMA protocol, more subcarriers are used for both cooperative and direct transmission modes when SNR increases. More subcarriers are used for cooperative mode when SNR is low, and more subcarriers are used for direct mode when SNR is greater than about 20 dB.

Fig. \ref{fig:power_portion} shows that the power percentages in the  proposed protocols. For the PS protocol, one observes that the power percentage of information transfer is decreasing with the total power constraint of the source node, and the power percentage of harvested energy at the relay node is decreasing with the source power. For the TMA protocol, the power percentage of cooperative mode is decreasing with the transmit power, and that of direct mode is increasing with the transmit power. Moreover, the transmit power used for information transfer (cooperative mode) is more than for power transfer (direct mode) when SNR is lower than about 20 dB, and the case is reversed when SNR is greater than about 20 dB.

\begin{figure}[t]
\begin{centering}
\includegraphics[scale=.6]{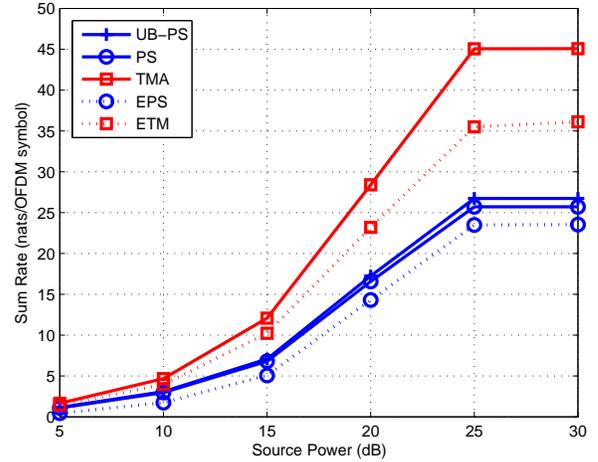}
\vspace{-0.1cm}
 \caption{Sum-rate comparison of the two proposed protocols, where the relay is located at the middle between the source and destination.}\label{fig:sum-rate}
\end{centering}
\vspace{-0.3cm}
\end{figure}
\begin{figure}[t]
\begin{centering}
\includegraphics[scale=.6]{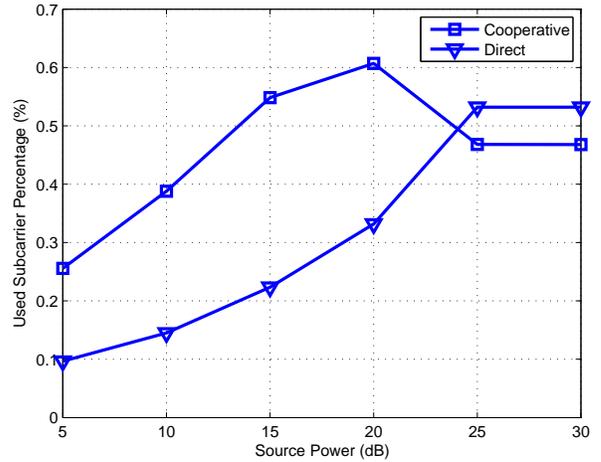}
\vspace{-0.1cm}
 \caption{Percentages of used subcarriers for the proposed TMA protocol, where the source power is fixed as $P=20$ dB.}\label{fig:sub_num}
\end{centering}
\vspace{-0.3cm}
\end{figure}
\begin{figure}[t]
\begin{centering}
\includegraphics[scale=.6]{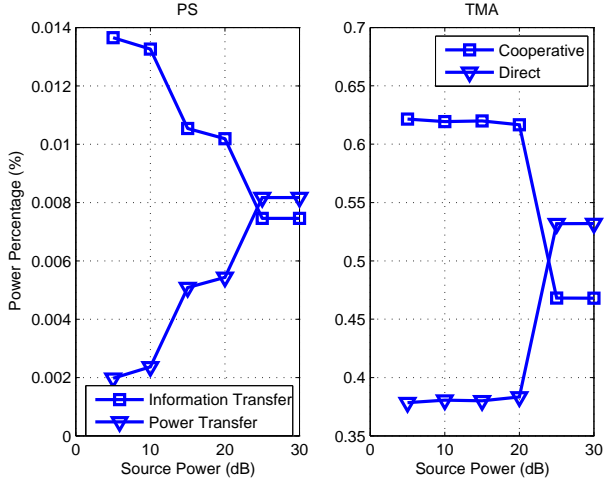}
\vspace{-0.1cm}
 \caption{Power percentages of the two protocols, where the source power is fixed as $P=20$ dB.}\label{fig:power_portion}
\end{centering}
\vspace{-0.3cm}
\end{figure}
%


\subsection{Impacts of Relay Location}

In this subsection, we consider the impacts of the relay location for system performance. Here, we fix the source power as $P=20$ dB, and let $d$ as the distance of the source node and the relay node.

In Fig. \ref{fig:sum-rate-relay}, we compare the system throughput of the protocols with different relay locations. It is observed that the TMA protocol outperforms the PS protocol over a wide range of relay locations. In addition, both protocols perform better when the relay node is closer to the source node. Specifically, the TMA protocol has the best performance when the relay node is located at $d=0.3$ and the PS protocol has the best performance when the relay node at about $d=0.5$.

From Fig. \ref{fig:power_portion_relay}, for the PS protocol, we observe that more power is consumed by power transfer when the relay location is closer to the source node, and more power is consumed by information transfer when the relay location is closer to the destination node. For TMA protocol, most power is used for direct mode when the relay location is closer to the source node, and the power percentage is decreasing with the distance of source-to-relay link. Moreover, the direct mode consumes minimum power when the relay is located at the middle between the source and destination.


%
\begin{figure}[t]
\begin{centering}
\includegraphics[scale=.6]{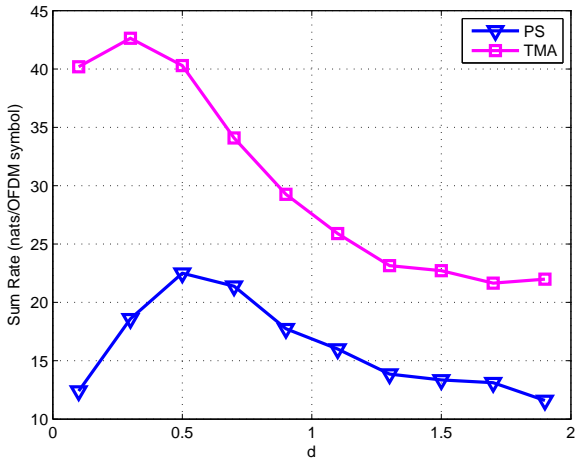}
\vspace{-0.1cm}
 \caption{Sum-rate comparison of the two protocols, where the source power is fixed as $P=20$ dB.}\label{fig:sum-rate-relay}
\end{centering}
\vspace{-0.3cm}
\end{figure}
\begin{figure}[t]
\begin{centering}
\includegraphics[scale=.6]{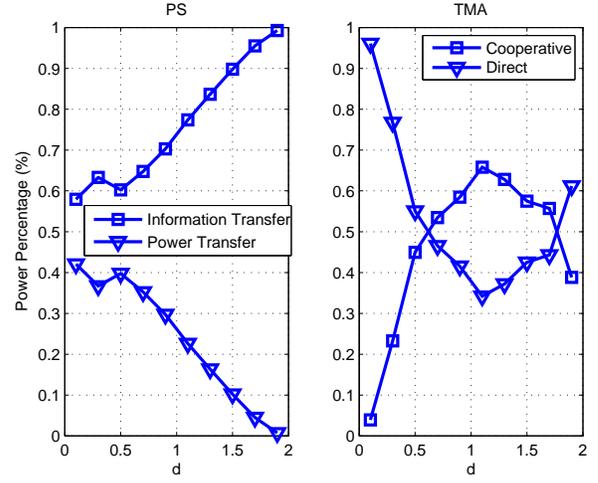}
\vspace{-0.1cm}
 \caption{Power percentages of the protocols, where the source power is fixed as $P=20$ dB.}\label{fig:power_portion_relay}
\end{centering}
\vspace{-0.3cm}
\end{figure}

\section{Conclusion}

In this paper, we studied the wireless information and energy cooperation in OFDM relaying systems. The key idea is that the source node transfers a fraction of power to the relay node in return for its assistance of information transmission. To support such simultaneous information cooperation and energy cooperation, two transmission protocols were proposed. The corresponding nonconvex resource allocation problems within the two protocols were solved with efficient algorithms.

\appendices

\section{Upper Bound of the Power Splitting Protocol}\label{app:UB}

Define $\mathcal D_3$ as the set of all primal variables that satisfy the constraints and $\bs\alpha=\{\alpha_n\}$, then the dual function can be expressed as
\begin{equation}
  g_{{\rm PS}}(\lambda,\mu,\nu)=\max_{\{\bs p,\bs\alpha\}\in\mathcal D_3}\mathcal L_{{\rm PS}}(\lambda,\mu,\nu,\bs p,\bs\alpha),
\end{equation}
where the Lagrangian is
\begin{align}\label{eqn:La-PS}
 \mathcal L_{{\rm PS}}(\lambda,\mu,\nu,\bs p,\bs\alpha)=\frac{\lambda}{2}\sum_{n\in\mathcal N}\log(1+p_{s,n}h_n\alpha_{n}) \nonumber\\
 +\frac{1-\lambda}{2}\sum_{n\in\mathcal N}\log(1+p_{r,n}g_n)
 +\mu\left[P-\sum_{n\in\mathcal N}p_{s,n}\right] \nonumber\\
 +\nu\left[\sum_{n\in\mathcal N}p_{s,n}h_n(1-\alpha_{n})-\sum_{n\in\mathcal N}p_{r,n}\right].
\end{align}

Note that in above we have eliminated the rate variable $R_{{\rm PS}}$ by using the similar method as in \eqref{eqn:g0} and $0\leq\lambda\leq1$. The maximization of $\mathcal L_{{\rm PS}}(\lambda,\mu,\nu,\bs p,\bs\alpha)$ can be decoupled into $2N$ subproblems, where $N$ subproblems for the first hop and $N$ subproblems for the second hop:
\begin{equation}\label{eqn:decomp-PS}
  \mathcal L_{{\rm PS}}(\lambda,\mu,\nu,\bs p,\bs\alpha)=\sum_{n\in\mathcal N}L_{{\rm PS},n}^s + \sum_{n\in\mathcal N}L_{{\rm PS},n}^r + \mu P,
\end{equation}
where
\begin{equation}\label{eqn:PSn1}
L_{{\rm PS},n}^s=\frac{\lambda}{2}\log(1+p_{s,n}h_n\alpha_{n})
-p_{s,n}\left[\mu-\nu h_n(1-\alpha_{n})\right],
\end{equation}
\begin{equation}
L_{{\rm PS},n}^r=\frac{1-\lambda}{2}\log(1+p_{r,n}g_n)-\nu p_{r,n}.
\end{equation}

We have the following proposition for deriving the optimal power allocation $\bs p^*$ and optimal power splitting ratio $\bs\alpha^*$.

\begin{proposition}\label{prop:p3}
The optimal solution to the performance upper bound is given by
\begin{enumerate}
  \item $\alpha_{n}^*=0$ and $p_{s,n}^*=\overline{p}_{n}$, if $\lambda\leq2\nu$ and $\nu h_n\geq\mu$,
  \item $\alpha_{n}^*=1$ and $p_{s,n}^*=\left[\frac{\lambda}{2\mu}-\frac{1}{h_n}\right]_0^{\overline{p}_{n}}$, if $\lambda>2\nu$ and $2\nu\overline{p}_{n}h_n\leq\lambda-2\nu$,
  \item $\alpha_{n}^*=\frac{\lambda-2\nu}{2\nu \overline{p}_{n}h_n}$ and $p_{s,n}^*=\overline{p}_{n}$, if $\lambda>2\nu$, $\nu h_n\geq\mu$, and $2\nu\overline{p}_{n}h_n>\lambda-2\nu$.
\end{enumerate}
For the relay,
\begin{equation}
p_{r,n}^*=\left[\frac{1-\lambda}{2\nu}-\frac{1}{g_n}\right]^{\overline{p}_{n}}_0.
\end{equation}
\end{proposition}

\begin{proof}
By taking the derivatives of $L_{{\rm PS},n}^s$ with respect to $p_{s,n}$ and $\alpha_{n}$, respectively, we have
\begin{equation}\label{eqn:partial-1}
\frac{\partial L_{{\rm PS},n}^s}{\partial p_{s,n}}=\frac{\lambda}{2}\cdot\frac{h_n\alpha_{n}}{1+p_{s,n}h_n\alpha_{n}}-[\mu-\nu h_n(1-\alpha_{n})],
\end{equation}
\begin{equation}\label{eqn:partial-2}
\frac{\partial L_{{\rm PS},n}^s}{\partial \alpha_{n}}=\frac{\lambda}{2}\cdot\frac{p_{s,n}h_n}{1+p_{s,n}h_n\alpha_{n}}
-p_{s,n}h_n\nu.
\end{equation}
%


Since $0\leq\alpha_{n}\leq1$, we have
\begin{eqnarray}
  \frac{\lambda}{2}\cdot\frac{p_{s,n}h_n}{1+p_{s,n}h_n}
-p_{s,n}h_n\nu\leq\frac{\partial L_{{\rm PS,n}}^s}{\partial \alpha_{n}}\nonumber\\
\leq\frac{\lambda p_{s,n}h_n}{2}-p_{s,n}h_n\nu.
\end{eqnarray}

We consider the following cases.
\begin{itemize}
\item \textbf{Case 1:} $\frac{\lambda p_{s,n}h_n}{2}-p_{s,n}h_n\nu\leq0$. In this case, $\lambda\leq2\nu$ and $L_{{\rm PS},n}$ is a monotonically decreasing function of $\alpha_{n}$. Thus we conclude that $\alpha_{n}^*=0$. Substituting $\alpha_{n}^*=0$ into \eqref{eqn:PSn1}, we have
    \begin{equation}
      L_{{\rm PS},n}^s=p_{s,n}(\nu h_n-\mu).
    \end{equation}
    To maximize $L_{{\rm PS},n}^s$, the optimal power allocation should be
    \begin{eqnarray}
      p_{s,n}^*=\begin{cases}
        \overline{p}_{n},&{\rm if}~\nu h_n\geq\mu\\
        0,&{\rm otherwise.}
      \end{cases}
    \end{eqnarray}

\item \textbf{Case 2:} $\frac{\lambda}{2}\cdot\frac{p_{s,n}h_n}{1+p_{s,n}h_n}
-p_{s,n}h_n\nu\geq0$. In this case, $\lambda>2\nu$ and $p_{s,n}\leq\frac{\lambda-2\nu}{2\nu h_n}$. Hence $L_{{\rm PS},n}^s$ is a monotonically increasing function of $\alpha_{n}$. This means that $\alpha_{n}^*=1$. Substituting $\alpha_{n}^*=1$ into \eqref{eqn:PSn1} and applying KKT conditions, the optimal power allocation has the form of water-filling:
\begin{equation}
p_{s,n}^*=\left[\frac{\lambda}{2\mu}-\frac{1}{h_n}\right]_0^{\overline{p}_{n}}.
\end{equation}
Combining the above expression with the conditions that $\lambda>2\nu$ and $p_{s,n}\leq\overline{p}_{n}\leq\frac{\lambda-2\nu}{2\nu h_n}$, we obtain $2\nu\overline{p}_{n}h_n\leq\lambda-2\nu$ must hold for this case.

\item \textbf{Case 3:} $\frac{\lambda}{2}\cdot\frac{p_{s,n}h_n}{1+p_{s,n}h_n}
-p_{s,n}h_n\nu<0$ and $\frac{\lambda p_{s,n}h_n}{2}-p_{s,n}h_n\nu>0$. In this case, $\lambda>2\nu$ and $p_{s,n}>\frac{\lambda-2\nu}{2\nu h_n}$, which yields to $2\nu\overline{p}_{n}h_n>\lambda-2\nu$. By letting \eqref{eqn:partial-2} be zero, we have
\begin{equation}
\alpha_{n}^*=\frac{\lambda-2\nu}{2\nu p_{s,n}h_n}.
\end{equation}
Substituting $\alpha_{n}^*$ into \eqref{eqn:PSn1}, we have
\begin{eqnarray}
  L_{{\rm PS},n}^s=p_{s,n}(\nu h_n-\mu)+\frac{\lambda}{2}\log\frac{\lambda}{2\nu}+\nu-\frac{\lambda}{2}.
\end{eqnarray}
To maximize $L_{{\rm PS},n}$, the optimal power allocation should be
\begin{eqnarray}
      p_{s,n}^*=\begin{cases}
        \overline{p}_{n},&{\rm if}~\nu h_n\geq\mu\\
        0,&{\rm otherwise.}
      \end{cases}
\end{eqnarray}
\end{itemize}

Combining the above results, Proposition \ref{prop:p3} is thus proved.

\end{proof}

Take a closer look at Proposition \ref{prop:p3}, it is found that if a subcarrier is used for pure information transfer, i.e., case 2), the optimal power allocation is also water-filling. On the other hand, if a subcarrier is used for pure or partial power transfer, i.e., cases 1) and 3), the optimal power allocation is a constant giving the channel gain is larger than or equal to a threshold $\mu/\nu$.

Then, the dual problem can be expressed as
\begin{subequations}
\begin{align}
\min_{\{\lambda,\mu,\nu\}}~&g_{\rm PS}(\lambda,\mu,\nu) \\
{\rm s.t.}~~&0\leq\lambda\leq1,\mu\geq0,\nu\geq0.
\end{align}
\end{subequations}

Similarly, we use the ellipsoid method to find the optimal dual variables $\{\lambda^*,\mu^*,\nu^*\}$ by using the following subgradients.
\begin{equation}\label{eqn:delta-lambda}
  \Delta\lambda = \frac{1}{2}\sum_{n\in\mathcal N}\log(1+p_{s,n}^*h_n\alpha_{n}^*)-\frac{1}{2}\sum_{n\in\mathcal N}\log(1+p_{r,n}^*g_n),
\end{equation}
\begin{equation}\label{eqn:delta-mu}
  \Delta\mu=P-\sum_{n\in\mathcal N}p_{s,n}^*,
\end{equation}
\begin{equation}\label{eqn:delta-nu}
  \Delta\nu=\sum_{n\in\mathcal N}p_{s,n}^*h_n(1-\alpha_{n}^*)-\sum_{n\in\mathcal N}p_{r,n}^*.
\end{equation}

%

The complexity of decomposition is $2N$.  Combining the complexity of the ellipsoid method, the total complexity of the algorithm is $\mathcal{O}(2Nq^2)$.

\bibliographystyle{IEEEtran}
\bibliography{IEEEabrv,havest}

\end{document}